\documentclass[runningheads]{llncs}
\usepackage{amssymb, latexsym, amsfonts, amsmath, lscape, amscd, epsfig, mathrsfs, tikz, enumerate}
\usepackage{array} 
\usepackage{multirow}
\usepackage[T1]{fontenc}

\usepackage[dvipsnames]{xcolor}

\usepackage{colortbl}

\usepackage{graphicx}

\usepackage{enumitem}

\usepackage{subcaption}
\usepackage{csquotes}
\usepackage[textwidth=\marginparwidth]{todonotes}
\usepackage{marginnote}
\usepackage{soul}
\usepackage{clrscode3e}
\usepackage[ruled,linesnumbered]{algorithm2e}
\usepackage{float}
\usepackage{cite}
\usepackage{apxproof}

\usepackage{thmtools}
\usepackage{thm-restate}

\usepackage{hyperref}
\usepackage{xcolor}
\hypersetup{
    colorlinks=true,
    linkcolor=blue,
    urlcolor=blue,
    citecolor=blue
}
\usepackage{cleveref}

\usepackage{mathtools}
\usepackage{pgf,tikz}
\usetikzlibrary{decorations.pathreplacing}
\usetikzlibrary{decorations.markings}
\usetikzlibrary{positioning, quotes, calc}
\usetikzlibrary{arrows}
\usetikzlibrary{backgrounds} 
\usepackage{tabularx, environ}
\usepackage{amsthm}

\spnewtheorem{observation}{Observation}{\bfseries}{\itshape}

\makeatletter

\newcolumntype{\expand}{}
\long\@namedef{NC@rewrite@\string\expand}{\expandafter\NC@find}

\NewEnviron{problembox}[2][]{%
  \def\problem@arg{#1}%
  \def\problem@framed{framed}%
  \def\problem@lined{lined}%
  \def\problem@doublelined{doublelined}%
  \ifx\problem@arg\@empty%
    \def\problem@hline{}%
  \else%
    \ifx\problem@arg\problem@doublelined%
      \def\problem@hline{\hline\hline}%
    \else%
      \def\problem@hline{\hline}%
    \fi%
  \fi%
  \ifx\problem@arg\problem@framed%
    \def\problem@tablelayout{|>{\bfseries}lX|c}%
    \def\problem@title{\multicolumn{2}{|%
      >{\raisebox{-\fboxsep}}%
      p{\dimexpr\textwidth-4\fboxsep-2\arrayrulewidth\relax}%
      |}{%
        \textsc{#2}%
      }}%
  \else
    \def\problem@tablelayout{>{\bfseries}lXc}%
    \def\problem@title{\multicolumn{2}{>%
      {\raisebox{-\fboxsep}}%
      p{\dimexpr\textwidth-4\fboxsep\relax}%
      }{%
        \textsc{#2}%
      }}%
  \fi%
  \bigskip\par\noindent%
  \begin{tabularx}{\textwidth}{\expand\problem@tablelayout}%
    \problem@hline%
    \problem@title\\[2\fboxsep]%
    \BODY\\\problem@hline%
  \end{tabularx}%
  \medskip\par%
}
\makeatother

\begin{document}

\title{On the Complexity of Hop Domination and 2-Step Domination in Graph Classes}

\titlerunning{Hop Domination and 2-Step Domination in Graph Classes}

\author{Sandip Das\inst{} \and
Sweta Das\inst{} \and
Sk Samim Islam\inst{}}

\authorrunning{S. Das, S. Das and S.S. Islam}

\institute{Indian Statistical Institute, Kolkata, India}




\maketitle
\section{Abstract}
The domination problem is a well-studied problem in graph theory. In this paper, we study two natural variants: the hop domination problem and the $2$-step domination problem. Let $G$ be a graph with vertex set $V$ and edge set $E$. For a graph $G$, a subset $S \subseteq V(G)$ is called an \emph{hop dominating set}  if every vertex not in $S$ lies at distance of exactly $2$ from at least one vertex in $S$. For $v\in V(G)$, let $N(v,2)$ denote the set of vertices in $V(G)$ that are at distance exactly $2$ from $v$. For a graph $G$, a subset $S \subseteq V(G)$ is called an \emph{$2$-step dominating set} if every vertex $v\in V(G)$ lies at a distance of exactly $2$ from at least one vertex in $S$. The \textsc{Hop Domination} (HD) problem and the \textsc{$2$-Step Domination} ($2$SD) problems ask whether a graph contains a hop domination set or a $2$-step domination set of size at most $k$, respectively. We study the computational complexity of these problems, and show that both are NP-complete, even when restricted to $d$-regular graphs for every $d\geq 3$, claw-free graphs and also unit disk graphs.

\section{Introduction}
In graph theory, a dominating set is a set of vertices such that every vertex in the
graph is either in the set or adjacent to at least one vertex of the set. Dominating sets model direct coverage problems and have applications in network design, surveillance, and resource allocation. For example, servers in communication networks, cameras in security systems, or facilities such as hospitals can be modeled as vertices of a dominating set to ensure every entity is immediately accessible.

Hop domination generalizes classical domination by requiring coverage at an exact multi-step distance,for example, at distance exactly two. This is useful in settings where interactions occur indirectly. In financial security, fraudulent activities often pass through intermediaries like money mules, which can be captured using hop domination. In computer networks, hop domination helps identify critical relay nodes so that all devices remain within controlled multi-step access. Thus, while domination captures immediate influence, hop domination is valuable for analyzing hierarchical or indirect interactions.
A related notion, called \emph{hop domination}, was introduced by 
Ayyaswamy et al.~\cite{natarajan2015hop}. 
A set $S \subseteq V(G)$ is a \emph{hop dominating set (HDS)} 
if every vertex in $V(G)\setminus S$ is at distance exactly two from some vertex in $S$. 
The minimum size of such a set is the \emph{hop domination number}, denoted $\gamma_h(G)$. 
A natural generalization, called \emph{$r$-hop domination} for $r \ge 2$, 
was later studied in~\cite{jalalvand2017complexity}.

For a graph $G$ and an integer $r \ge 2$, a set $S \subseteq V(G)$ 
is an \emph{$r$-hop dominating set} ($r$HDS) if every vertex in $V(G)\setminus S$ 
is at distance exactly $r$ from some vertex in $S$. 
The minimum cardinality of such a set is the 
\emph{$r$-hop domination number}, denoted $\gamma_{rh}(G)$. 
The $r$-Hop Domination problem asks whether $G$ 
admits an $r$-hop dominating set of size at most $k$.

The case $r=2$ corresponds to the classical \textsc{Hop Domination} problem. 
Henning et al.~\cite{henning20172} showed that \textsc{Hop Domination} 
is \textsc{NP}-complete for planar bipartite and planar chordal graphs. 
They further proved that computing a minimum hop dominating set 
in an $n$-vertex graph cannot be approximated within a factor of 
$(1-\varepsilon)\log n$ for any $\varepsilon>0$, unless $\textsc{P}=\textsc{NP}$, 
and provided a polynomial-time approximation algorithm with a ratio 
$1+\log\bigl(\Delta(\Delta-1)+1\bigr)$, where $\Delta$ is the maximum degree of $G$~\cite{henning2020algorithm}. 
Moreover, the problem is \textsc{APX}-complete for bipartite graphs of maximum degree $3$.

Karthika et al.~\cite{karthika2025polynomial} presented polynomial-time algorithms for the \textsc{Hop Domination} problem on interval graphs and biconvex bipartite graphs. They also investigated the parameterized complexity of the problem, proving that it is \textsc{W[1]}-hard when parameterized by the solution size, and subsequently strengthened this result to \textsc{W[2]}-hardness~\cite{karthika2025hop}. Farhadi et al.~\cite{jalalvand2017complexity} showed that, for every $r \geq 2$, the \textsc{$r$-Hop Domination} problem is \textsc{NP}-complete even on planar bipartite and planar chordal graphs.

For a graph $G$, a subset $S \subseteq V(G)$ is called a \textit{$2$-step dominating set} if every vertex $v \in V(G)$ is at distance exactly $2$ from at least one vertex in $S$. Henning et al.~\cite{henning20172} proved that the \textsc{$2$-Step Domination} problem is \textsc{NP}-complete for planar bipartite and chordal graphs. More recently, Das et al.~\cite{das2026parameterized} established that both the \textsc{$r$-Hop Domination} and \textsc{$r$-Step Domination} problems are \textsc{W[2]}-hard even on bipartite and chordal graphs.

\medskip
\noindent\textbf{Our Contributions.}
We prove the NP-completeness of the \textsc{Hop Domination} and \textsc{$2$-Step Domination} problems on several graph classes. First, by a reduction from the \textsc{Planar Vertex Cover} problem restricted to graphs of maximum degree~$3$, which is known to be NP-complete~\cite{garey1977rectilinear}, we establish the hardness of both problems on unit disk graphs.

\begin{restatable}[]{theorem}{HopTwoStepUnitDisk}\label{thm:hop-2step-unitdisk-npc}
\textsc{Hop Domination} and \textsc{$2$-Step Domination} problems are NP-complete on unit disk graphs.
\end{restatable}

We then study these problems on regular graphs and obtain the following result.

\begin{restatable}[]{theorem}{HopTwoStepRegular}\label{thm:hop-2step-regular-npc}
\textsc{Hop Domination} and \textsc{$2$-Step Domination} problems are NP-complete on regular graphs.
\end{restatable}

Finally, we show that the hardness persists even for claw-free graphs.

\begin{restatable}[]{theorem}{HopTwoStepClawFree}\label{thm:hop-2step-clawfree-npc}
\textsc{Hop Domination} and \textsc{$2$-Step Domination} problems are NP-complete on claw-free graphs.
\end{restatable}

\medskip
\noindent\textbf{Organisation:} In Section \ref{sec:prelim}, we recall some definitions and notations.  In Section \ref{sec:unitdiskgraph}, we show that \textsc{Hop Domination} and \textsc{$2$-Step Domination} problems are NP-complete on unit disk graphs. In Sections \ref{sec:3regular} and \ref{sec:dregular}, we show the NP-completeness of the problems for regular graphs. In Section \ref{sec:clawfree}, we show that the problems remain NP-complete for claw-free graphs also. We conclude in Section \ref{sec:conclusion}.

\section{Preliminaries}\label{sec:prelim}

In this paper, we consider finite, simple, and undirected graphs. 
For a graph $G=(V,E)$ and a vertex $v\in V(G)$, the \emph{open neighborhood} of $v$ is denoted by 
$N_G(v)=\{u\in V(G)\mid uv\in E(G)\}$, and the \emph{closed neighborhood} of $v$ is 
$N_G[v]=N_G(v)\cup \{v\}$. The degree of a vertex $v$ is denoted by $d_G(v)$. 
For two vertices $u,v\in V(G)$, the distance between $u$ and $v$ in $G$ is denoted by $d_G(u,v)$.

A graph $G$ is called a \emph{unit disk graph} if there exists a representation of the vertices of $G$ as points in the Euclidean plane such that two vertices are adjacent in $G$ if and only if the Euclidean distance between the corresponding points is at most $1$.

A graph $G$ is called \emph{$r$-regular} if every vertex of $G$ has degree exactly $r$. A graph is said to be \emph{regular} if it is $r$-regular for some positive integer $r$.

The graph $K_{1,3}$ is called a \emph{claw}. A graph $G$ is said to be \emph{claw-free} if $G$ does not contain an induced subgraph isomorphic to $K_{1,3}$.

 \section{Unit Disk Graphs}\label{sec:unitdiskgraph}

In this section, we prove that \textsc{Hop Domination} problem and \textsc{$2$-Step Domination} problem are NP-complete in unit disk graphs.

\subsection{Hop Domination}

We sketch a polynomial time transformation to the given problem from 
\textsc{Planar Vertex Cover} restricted to graphs of maximum degree~$3$, 
which is shown to be NP-complete in~\cite{garey1977rectilinear}.

\medskip
\noindent
\fbox{%
  \begin{minipage}{\dimexpr\linewidth-2\fboxsep-2\fboxrule}
  \textsc{Vertex Cover} problem
  \begin{itemize}[leftmargin=1.5em]
    \item \textbf{Input:} An undirected graph $G = (V,E)$ and an integer $k \in \mathbb{N}$.
    \item \textbf{Question:} Does there exist a \emph{vertex cover} $C \subseteq V$ of size at most $k$; that is, a set of at most $k$ vertices such that every edge in $E$ has at least one endpoint in $C$?
  \end{itemize}
  \end{minipage}%
}
\medskip

We transform a planar graph $G_1=(V,E)$ with maximum degree~$3$ into a 
unit disk graph $G_2$ such that $G_1$ has a vertex cover of size at 
most~$k$ if and only if $G_2$ has a hop dominating set of size at 
most~$k+k'$, where $k'$ is a constant determined by the construction.  We construct an intersection model for $G_2$ by making use of the 
following result from~\cite{valiant1981universality}:

\begin{lemma}[Valiant~\cite{valiant1981universality}]\label{lem:planarEmbedding}
A planar graph $G$ with maximum degree~$4$ can be embedded in the plane 
using $O(|V|)$ area in such a way that its vertices are at integer 
coordinates and its edges are drawn so that they are made up of line 
segments of the form $x=i$ or $y=j$, for integers $i$ and $j$ See~Fig.\ref{fig:planar embedding}.
\end{lemma}

 \begin{figure}[h]
    \centering
   \includegraphics[width=0.75\textwidth]{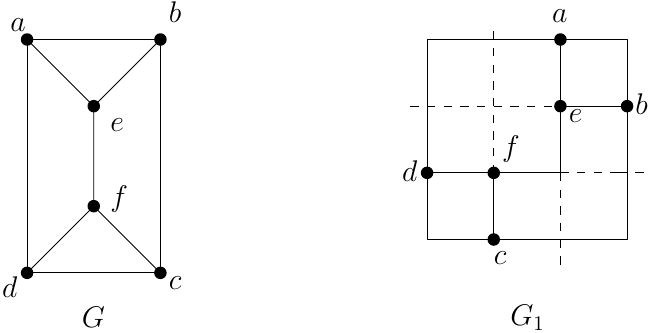}
    \caption{Planar embedding of a planar graph with max degree $4$}
    \label{fig:planar embedding}
\end{figure}

Algorithms to produce such embeddings efficiently are given for 
example in~\cite{clark1990unit, itai1982hamilton}. Using one of them 
we construct such an embedding of $G_1$, adjusting the scale so that 
the horizontal and vertical straight line segments that make up the 
edges are each of integer length~$8$. The vertices of $G_1$ are 
modeled by circles of radius~$1/2$ centered at the locations of the 
vertices in this embedding. The edges of $G_1$ are replaced by chains 
of radius-$1/2$ circles, augmented with gadgets attached at each 
interior integer coordinate point along the edge, specified as 
follows. If $e_i$ is the edge between vertices $u$ and $v$, then the 
set of circles used to represent it consists of a primary chain 
$C[e_i]=\{c_{i_1},c_{i_2},\dots,c_{i_{k_i}}\}$, where $k_i$ depends 
on the length of the embedding of $e_i$, together with a collection 
of attached gadgets, one at each interior integer coordinate point of 
$e_i$. The chain circles are positioned along the embedded segments 
of $e_i$ so that consecutive circles intersect, yielding an 
intersection pattern like that shown in Fig.~\ref{fig:unitdiskhop} 
(top left) for an edge made of a single horizontal line segment. In 
addition, at every integer coordinate point $d$ lying in the interior 
of the embedded edge $e_i$, let $c_{d_1}$ and $c_{d_2}$ denote the 
two chain circles of $C[e_i]$ positioned immediately to the left and 
right of $d$, respectively. We attach to the chain at $d$ a 
two-circle gadget consisting of radius-$1/2$ circles $c_{d_3}$ and 
$c_{d_4}$ placed perpendicular to the chain, with $c_{d_3}$ 
positioned so that it intersects each of $c_{d_1}$, $c_{d_2}$, and 
$c_{d_4}$, while $c_{d_4}$ intersects only $c_{d_3}$. We denote the 
resulting four-circle configuration at $d$ by 
$D[d]=\{c_{d_1},c_{d_2},c_{d_3},c_{d_4}\}$. Finally, every chain 
circle $w \in C[e_i]$ that intersects either a vertex circle (namely 
$u$ or $v$) or one of the flanking gadget circles $c_{d_1}, c_{d_2}$ 
(for some interior integer coordinate point $d$) has attached to it 
a single pendant radius-$1/2$ circle $p_w$, positioned so that $p_w$ 
intersects only $w$ and no other circle of the construction, as 
illustrated by the pendant circles drawn near $u$ and $v$ in 
Fig.~\ref{fig:unitdiskhop} (top right). The resulting intersection 
pattern is illustrated in Fig.~\ref{fig:unitdiskhop} (top right). 
The reader may verify that this representation for a horizontal edge 
can be modified to bend around corners if the edge to be represented 
consists of both horizontal and vertical segments, given that each 
segment by construction is of length~$8$. Applying this construction 
to every edge of $G_1$ yields the unit disk graph $G_2$, as 
illustrated in Fig.~\ref{fig:unitdiskhop} (bottom).

\begin{figure}[h]
    \centering
   \includegraphics[width=0.95\textwidth]{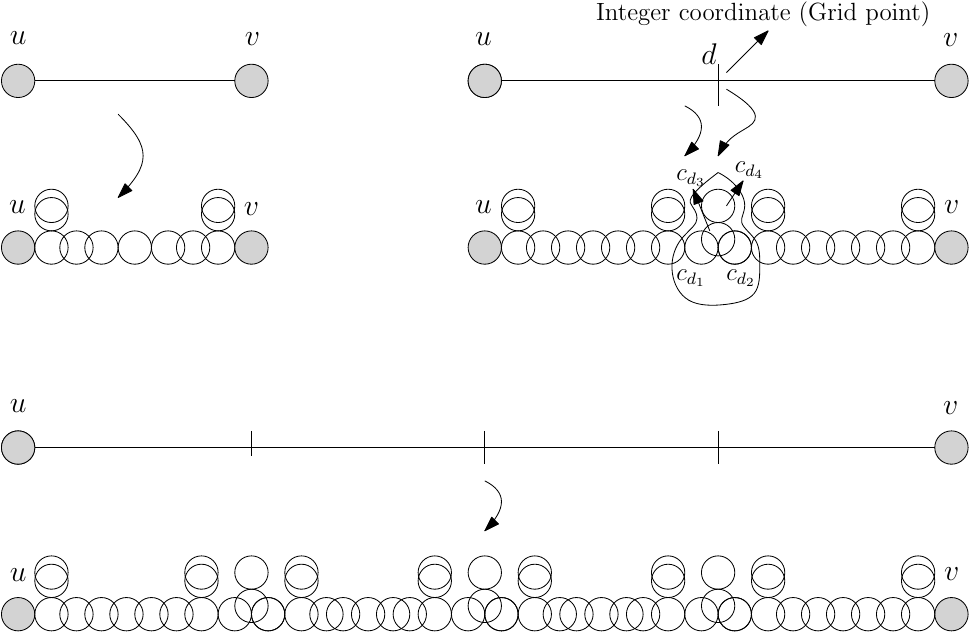}
    \caption{Gadget construction of $G_2$ from $G_1$}
    \label{fig:unitdiskhop}
\end{figure}

Thus, $G$ has a vertex cover of size $k$ if and only if $G_2$ has a hop dominating set of size $k + \sum_{(u,v)\in E(G)}(4k_{uv} -1)$, where $k_{uv}$ is the length of an edge between $u$ and $v$ in $G_1$.

Hence, the \textsc{Hop Domination} problem is NP-complete in unit disk graphs.


\subsection{$2$-Step Domination}

We sketch a polynomial time transformation to the \textsc{$2$-Step Domination} problem from 
\textsc{Planar Vertex Cover} restricted to graphs of maximum degree~$3$, 
which is shown to be NP-complete in~\cite{garey1977rectilinear}. 
We transform a planar graph $G_1=(V,E)$ with maximum degree~$3$ into a 
unit disk graph $G_2$ such that $G_1$ has a vertex cover of size at 
most~$k$ if and only if $G_2$ has a $2$-step dominating set of size 
at most~$k+k'$, where $k'$ is a constant determined by the 
construction. We construct an intersection model for $G_2$ by making 
use of Lemma \ref{lem:planarEmbedding}
\begin{figure}[]
    \centering
   \includegraphics[width=0.95\textwidth]{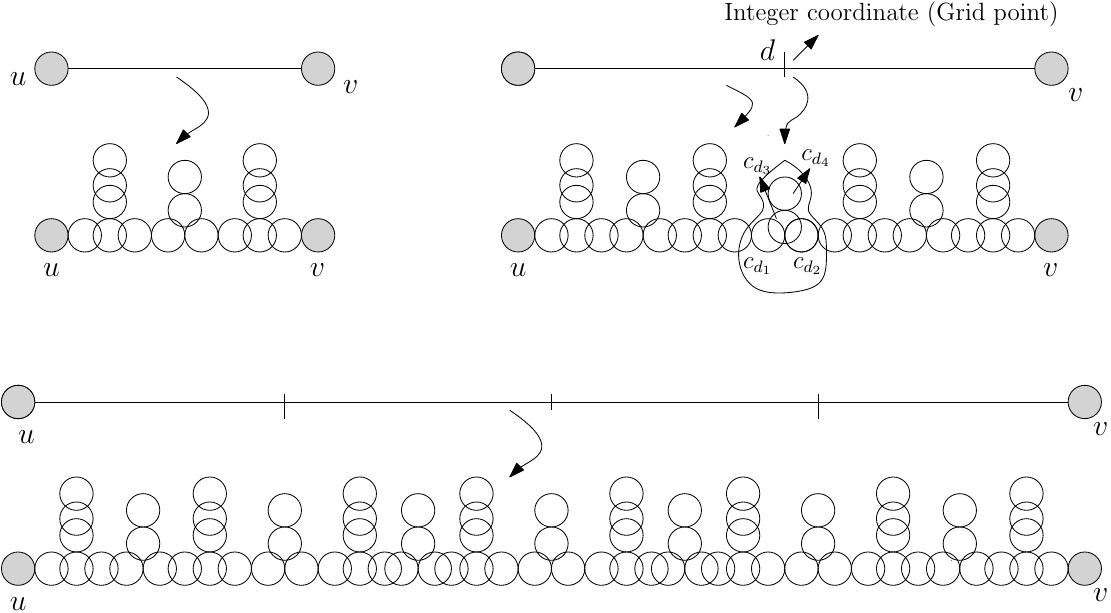}
    \caption{Gadget construction of $G_2$ from $G_1$}
    \label{fig: unit disk $2$-step}
\end{figure}

Efficient algorithms for constructing the embeddings described in Lemma~\ref{lem:planarEmbedding} are given, for example, in~\cite{clark1990unit, itai1982hamilton}. Using one of them 
we construct such an embedding of $G_1$, adjusting the scale so that 
the horizontal and vertical straight line segments that make up the 
edges are each of integer length~$8$. The vertices of $G_1$ are 
modeled by circles of radius~$1/2$ centered at the locations of the 
vertices in this embedding. The edges of $G_1$ are replaced by chains 
of radius-$1/2$ circles, augmented with gadgets attached at each 
interior integer coordinate point along the edge and with boundary 
gadgets attached near each endpoint, specified as follows. If $e_i$ 
is the edge between vertices $u$ and $v$, then the set of circles 
used to represent it consists of a primary chain 
$C[e_i]=\{c_{i_1},c_{i_2},\dots,c_{i_{k_i}}\}$, where $k_i$ depends 
on the length of the embedding of $e_i$, together with a collection 
of attached gadgets. The chain circles are positioned along the 
embedded segments of $e_i$ so that consecutive circles intersect, 
yielding an intersection pattern like that shown in 
Fig.~\ref{fig: unit disk $2$-step} (top left) for an edge made of a single 
horizontal line segment. At every integer coordinate point $d$ lying 
in the interior of the embedded edge $e_i$, let $c_{d_1}$ and 
$c_{d_2}$ denote the two chain circles of $C[e_i]$ positioned 
immediately to the left and right of $d$, respectively. We attach 
to the chain at $d$ a two-circle gadget consisting of radius-$1/2$ 
circles $c_{d_3}$ and $c_{d_4}$ placed perpendicular to the chain, 
with $c_{d_3}$ positioned so that it intersects each of $c_{d_1}$, 
$c_{d_2}$, and $c_{d_4}$, while $c_{d_4}$ intersects only $c_{d_3}$. 
We denote the resulting four-circle configuration at $d$ by 
$D[d]=\{c_{d_1},c_{d_2},c_{d_3},c_{d_4}\}$. In addition, in the 
segment of $e_i$ between vertex $v$ and the rightmost interior 
integer coordinate point of $e_i$ (and symmetrically in the segment 
between vertex $u$ and the leftmost interior integer coordinate 
point), a three-circle boundary gadget is attached to the chain. 
Specifically, let $w$ and $w'$ denote two distinct chain circles of 
$C[e_i]$ in this segment, with $w$ positioned closer to the 
rightmost interior integer coordinate point and $w'$ positioned 
closer to $v$. We attach to $w$ a pair of radius-$1/2$ circles 
$b^v_1$ and $b^v_2$ placed perpendicular to the chain, with $b^v_1$ 
positioned so that it intersects both $w$ and $b^v_2$, while $b^v_2$ 
intersects only $b^v_1$. We further attach to $w'$ a single 
radius-$1/2$ circle $b^v_3$ placed perpendicular to the chain, 
intersecting only $w'$. We denote the resulting three-circle 
boundary configuration near $v$ by 
$B[v]=\{b^v_1,b^v_2,b^v_3\}$, and define $B[u]$ symmetrically. The 
resulting intersection pattern is illustrated in 
Fig.~\ref{fig: unit disk $2$-step} (top right), with the circles of $B[v]$. The reader may verify that this 
representation for a horizontal edge can be modified to bend around 
corners if the edge to be represented consists of both horizontal and 
vertical segments, given that each segment by construction is of 
length~$8$. Applying this construction to every edge of $G_1$ yields 
the unit disk graph $G_2$, as illustrated in 
Fig.~\ref{fig: unit disk $2$-step} (bottom).

Thus, $G$ has a vertex cover of size $k$ if and only if $G_2$ has a $2$-step dominating set of size $k + \sum_{(u,v)\in E(G)}(8k_{uv} - 1)$, where $k_{uv}$ is the length of an edge between $u$ and $v$ in $G_1$.

Hence, the \textsc{$2$-Step Domination} problem is NP-complete in unit disk graphs.


\section{$3$-Regular graph}\label{sec:3regular}
In this section, we prove that \textsc{Hop Domination} problem and \textsc{$2$-Step Domination} problem are NP-complete for $3$-regular graphs.
\subsection{Hop Domination}
In this subsection, we show that \textsc{Hop Domination} problem in a $3$-regular graph is NP-complete using the fact that the vertex cover problem in a $3$-regular graph is NP-complete\cite{alimonti2000some}. We reduce vertex cover problem to our problem.
Given a $3$-regular graph $G_1 = (V(G_1), E(G_1))$ with vertex set $V(G_1) = \{v_1, v_2, \dots, v_n\}$, we construct a new graph $G_2$ as follows:
\begin{enumerate}
    \item  First, we take the set of vertex $\{u_1, u_2, \dots, u_n\}$, where each $u_i \in V(G_2)$ corresponds to the vertex $v_i \in V(G_1)$.
       \item For each edge $(v_i,v_j)\in E(G_1)$, where $i < j$, we add a new vertex $u_{ij}$ and join with $u_i$ and $u_j$. 
       \item For each $u_{ij}$, we introduce five new vertices $a_{ij},b_{ij},c_{ij},d_{ij},e_{ij}$ and consider the edges $(a_{ij},u_{ij}),(a_{ij},b_{ij}),(a_{j},c_{ij}),(b_{ij},c_{ij}),(b_{ij},d_{ij}),(c_{ij},e_{ij})$. 
       \item For each $d_{ij}$, we introduce $12$ new vertices $d^{1k}_{ij}$ and $d^{2k}_{ij}$, where $k=1,2,\dots 6$.
       \item At each $d_{ij}$, we introduce two $6$ length path i.e, $d_{ij}d_{ij}^{11}d_{ij}^{12}\dots d_{ij}^{16}$ and $d_{ij}d_{ij}^{21}d_{ij}^{22}\dots d_{ij}^{26}$.
       \item We consider the edges $(d_{ij}^{1k},d_{ij}^{2k})$ where $k=1,2,3,4$.
       \item Then we consider the edges $(d_{ij}^{25},d_{ij}^{16}),(d_{ij}^{15},d_{ij}^{26})$.
        \item For each $e_{ij}$, we introduce $12$ new vertices $e^{1k}_{ij}$ and $e^{2k}_{ij}$, where $k=1,2,\dots 6$.
       \item At each $e_{ij}$, we introduce two $6$ length path i.e, $e_{ij}e_{ij}^{11}e_{ij}^{12}\dots e_{ij}^{16}$ and $e_{ij}e_{ij}^{21}e_{ij}^{22}\dots e_{ij}^{26}$.
       \item We introduce the edges $(e_{ij}^{1k},e_{ij}^{2k})$ where $k=1,2,3,4$.
       \item We introduce the edges $(e_{ij}^{25},e_{ij}^{16}),(e_{ij}^{15},e_{ij}^{26})$.
       
        \textcolor{red}.
\end{enumerate}
\begin{figure}[h]
    \centering
   \includegraphics[width=\textwidth]{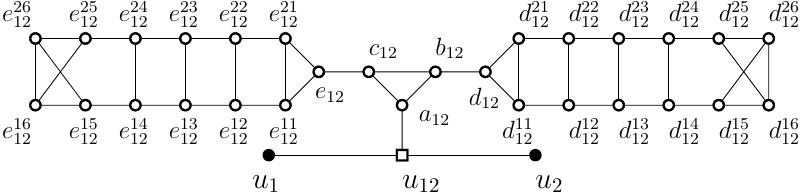}
    \caption{Gadget construction of $G_2$ from $G_1$}
    \label{}
\end{figure}
    \begin{lemma}\label{lemma:k free hop}
           $G_1$ has a vertex cover of size at most $k$ if and only if $G_2$ has a hop dominating set of size at most $k+6m$, where $m=|E(G_1)|$.
       \end{lemma}
       \begin{proof}
        Let us assume $S_{G_1}$ is a vertex cover of size at most $k$. Now we take $S_{G_2} =S_{G_1} \cup \ \{b_{ij},c_{ij},d^{24}_{ij},d^{23}_{ij},e^{24}_{ij},e^{23}_{ij}\}$, which is a hop dominating set. Conversely, let $G_2$ have a hop dominating set $S_{G_2}$ of size at most $k+6m$. We denote the induced subgraph of $G_2[\{b_{ij},d_{ij},d_{ij}^{1k},d_{ij}^{2k}\}]$, where $k=1,2,\dots,6$ as $B$ and also the induced subgraph of $G_2[\{c_{ij},e_{ij},e_{ij}^{1k},e_{ij}^{2k}\}]$, where $k=1,2,\dots,6$ as $C$. Now to hop dominate $d^{25}_{ij},d^{26}_{ij},d^{15}_{ij},d^{16}_{ij}$ these four vertices, at least two vertices belong to $S_{G_2}$. Let $A\subseteq S_{G_2}$, hop dominate $d^{25}_{ij},d^{26}_{ij},d^{15}_{ij},d^{16}_{ij}$ where $|A|=2$. Notice that at least one of $d^{21}_{ij},d^{11}_{ij}$ these vertices can not hop dominate by $A$. To hop dominate $d^{21}_{ij},d^{11}_{ij}$ these vertices we can take $b_{ij}$ in $S_{G_2}$.  Hence, To hop dominate $B$ at least $3$ vertices belongs to $S_{G_2}$. similarly to hop dominate $C$ at least $3$ vertices belong to $S_{G_2}$. To hop dominate $e^{21}_{ij},e^{11}_{ij}$ these vertices we can take $c_{ij}$ in $S_{G_2}$.  To hop dominate $d^{21}_{ij},d^{11}_{ij}$ these vertices we can take $b_{ij}$ in $S_{G_2}$. Hence, for each edge, at least $6$ vertices belong to $S_{G_2}$. If $e_{ij},a_{ij},d_{ij}$ are not in $S_{G_2}$, to hop dominate $a_{ij}$ at least one of $u_i$ or $u_j$ belongs to $S_{G_2}$. If some of $e_{ij},a_{ij},d_{ij}$ in $S_{G_2}$, we can delete these vertices from $S_{G_2}$ and  one of $u_i$ or $u_j$ or both included in $S_{G_2}$. Notice that, the set $\{v_l:u_l\in S_{G_2}$\} is a vertex cover of $G_1$ of size $\leq k$. 
     \end{proof}
Hence the \textsc{Hop Domination} problem is NP-complete for $3$-regular graphs.
\subsection{$2$-Step Domination}
     In this subsection, we show that \textsc{$2$-Step Domination} problem in a $3$-regular graph is NP-complete. We reduce the vertex cover problem to our problem.
Given a $3$-regular graph $G_1 = (V(G_1), E(G_1))$ with vertex set $V(G_1) = \{v_1, v_2, \dots, v_n\}$, we construct a new graph $G_2$ as follows:
\begin{enumerate}
    \item  First, we take the set of vertex $\{u_1, u_2, \dots, u_n\}$, where each $u_i \in V(G_2)$ corresponds to the vertex $v_i \in V(G_1)$.
       \item For each edge $(v_i,v_j)\in E(G_1)$, where $i < j$, we add a new vertex $u_{ij}$ and join with $u_i$ and $u_j$. 
       \item For each $u_{ij}$, we introduce three new vertices $a_{ij},b_{ij},c_{ij}$ and consider the edges $(a_{ij},u_{ij}),(a_{ij},b_{ij}),(a_{j},c_{ij})$. 
       \item For each $b_{ij}$, we introduce $4$ new vertices $b^{11}_{ij},b^{12}_{ij}$ and $b^{21}_{ij},b^{22}_{ij}$.
       \item At each $c_{ij}$, we introduce $4$ new vertices $c^{11}_{ij},c^{12}_{ij}$ and $c^{21}_{ij},c^{22}_{ij}$.
       \item We introduce the edges $(b_{ij}^{11},b_{ij}^{22}),(b_{ij}^{21},b_{ij}^{12})$,$(c_{ij}^{11},c_{ij}^{22}),(c_{ij}^{21},c_{ij}^{12})$.
\end{enumerate}
\begin{figure}[h]
    \centering
   \includegraphics[width=0.55\textwidth]{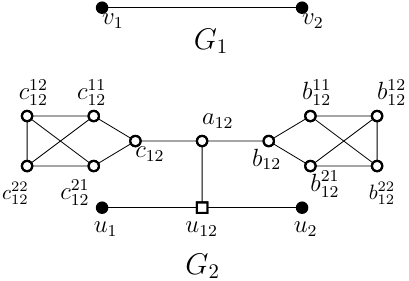}
    \caption{Gadget construction of $G_2$ from $G_1$}
    \label{}
\end{figure}
 \begin{lemma}\label{lemma:k free hop}
           $G_1$ has a vertex cover of size at most $k$ if and only if $G_2$ has a $2$-step dominating set of size at most $k+3m$, where $m=|E(G_1)|$.
       \end{lemma}
       \begin{proof}
        Let us assume $S_{G_1}$ is a vertex cover of size at most $k$. Now we take $S_{G_2} =S_{G_1} \cup \ \{a_{ij},b_{ij},c_{ij}\}$, which is a $2$-step dominating set. Conversely, let $G_2$ have a $2$-step dominating set $S_{G_2}$ of size at most $k+3m$. To hop dominate $c^{12}_{ij}$ and $c^{22}_{ij}$ $c_{ij}$ must be in $S_{G_2}$, similarly, to hop dominate $b^{12}_{ij}$ and $b^{22}_{ij}$ $b_{ij}$ must be in $S_{G_2}$. Now, notice that till now $b^{11}_{ij}$ and $b^{21}_{ij}$ and $c^{11}_{ij}$ and $c^{21}_{ij}$ are not hop dominated, thus for every edge $(v_i,v_j)\in E(G_1)$ at least $3$ vertices from $V(G_2)$ are in $S_{G_2}$. To hop dominate $b^{11}_{ij}$ and $b^{21}_{ij}$ and $c^{11}_{ij}$ and $c^{21}_{ij}$ we can take $a_{ij}$ in $S_{G_2}$. Now to hop dominate any vertex $\in a_{ij}$, either $u_i$ or $u_j$ must be in $S_{G_2}$. Notice that, the set $\{v_l:u_l\in S_{G_2}$\} is a vertex cover of $G_1$ of size $\leq k$.

        Hence the \textsc{$2$-Step Domination} problem is NP-complete for $3$-regular graphs.
     \end{proof}
\section{$d$-Regular Graphs, for $d\geq4$} \label{sec:dregular}
In this section, we prove that \textsc{Hop Domination} problem and \textsc{$2$-Step Domination} problem are NP-complete for $d$-regular graphs, where $d\geq4$.
From the Erdős-Gallai Theorem~\cite{erdos1960grafok} and the Havel-Hakimi criterion~\cite{hakimi1962realizability,havel1955remark}, we can say the following.

\begin{lemma}[\cite{erdos1960grafok,hakimi1962realizability,havel1955remark}]\label{lem:ErdosHavelHakimi}\label{lem:havel_hakimi}
    A simple $d$-regular graph on $n$ vertices can be formed in polynomial time when  $n \cdot d$ is even and $d < n$.
\end{lemma}

We use Lemma \ref{lem:ErdosHavelHakimi} to prove the NP-completeness of the \textsc{Hop Domination} and \textsc{$2$-Step Domination} problem on $d$-regular graphs, where $d\geq4$.

\subsection{Hop Domination}
  In this subsection, we show that \textsc{Hop Domination} problem is NP-complete in $d$-regular graphs, where $d\geq4$. It is known that \textsc{Vertex Cover} problem is NP-complete for $d$-regular graphs~\cite{garey2002computers}.  We reduce this problem to our problem.
Given a $d\geq 4$-regular graph $G_1 = (V(G_1), E(G_1))$ with vertex set $V(G_1) = \{v_1, v_2, \dots, v_n\}$, we construct a new graph $G_2$ as follows:
\begin{figure}[h]
    \centering
   \includegraphics[width=0.9\textwidth]{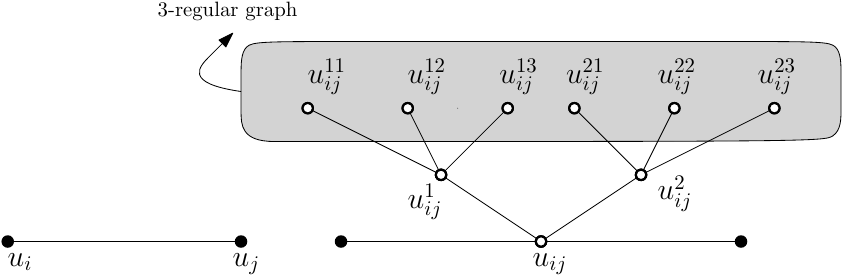}
    \caption{Gadget construction of $G_2$ from $G_1$, for example $4$ regular graph}
    \label{}
\end{figure}
\begin{enumerate}
    \item First, we take the set of vertices $\{u_1, u_2, \dots, u_n\}$, where each $u_i \in V(G_2)$ corresponds to the vertex $v_i \in V(G_1)$.
       \item For each edge $(v_i,v_j)\in E(G_1)$, where $i < j$, we add a new vertex $u_{ij}$ and join it with $u_i$ and $u_j$.
       \item For each $u_{ij}$, we introduce $(d-2)$ new vertices $u^1_{ij},u^2_{ij},\dots,u^{d-2}
       _{ij}$ and join each of them with $u_{ij}$.
       \item For each $u^k_{ij}$, we introduce $(d-1)$ new vertices $u^{k1}_{ij},u^{k2}_{ij},\dots,u^{k(d-1)}_{ij}$.
       We denote the induced the $H_{ij}:=G_2[\{u_i,u_j,u_{ij}\} \cup \{u_{ij}^k,u_{ij}^{km}:k=1,2\dots,d-2; m=1,2,\dots,d-1\}]$, $H_{ij}^0=G_2[\{u_i,u_j,u_{ij}\}]$, $H_{ij}^1=G_2[\{u_{ij}^{k}:l=1,2,\dots,d-2\}]$, $H_{ij}^2=G_2[\{u_{ij}^{km}:k=1,2,\dots,d-2,m=1,2,\dots,d-1,\}]$, therefore $H_{ij}=H_{ij}^0\cup H^1_{ij}\cup H^2_{ij}$.
       \item Using Lemma \ref{lem:havel_hakimi}, we form a $d$-regular graph with the vertices of $H^2_{ij}$, for every $u_{ij}$.
\end{enumerate}
 \begin{lemma}\label{lemma:d regular hop}
           $G_1$ has a vertex cover of size at most $k$ if and only if $G_2$ has a hop dominating set of size at most $k+m$, where $m=|E(G_1)|$.
       \end{lemma}
       \begin{proof}
             Let us assume $S_{G_1}$ is a vertex cover of size at most $k$. Now we take $S_{G_2} =S_{G_1} \cup_{(v_i,v_j)\in E(G_1)} \{u_{ij}\}$, which is a hop dominating set. Conversely, let $G_2$ have a hop dominating set $S_{G_2}$ of size at most $k+m$. Let $x\in H^2_{ij}$ hop dominated by a vertex $y$, now $y$ belong to $H_{ij}^2$ or $y=u_{ij}$, notice that there exist at least one vertex in $H_{ij}^1$ and one vertex in $H_{ij}^2$ can not be hop dominated by this $y$. Hence, 
             if $y\in H_{ij}^2$, we can replace it with by $u_{ij}$. Now to hop dominate any vertex $\in H_{ij}^0$, either $u_i$ or $u_j$ must be in $S_{G_2}$. Notice that, the set $\{v_l:u_l\in S_{G_2}$\} is a vertex cover of $G_1$ of size $\leq k$. \hfill $\square$
             
       \end{proof}

Hence the \textsc{Hop Domination} problem is NP-complete for $d$-regular graphs, where $d\geq4$. 
       
       \subsection{$2$-Step Domination}
In this subsection, we show that \textsc{$2$-Step Domination} problem is NP-complete in $d$-regular graphs, where $d\geq4$. It is known that \textsc{Vertex Cover} problem is NP-complete for $d$-regular graphs~\cite{garey2002computers}.  We reduce this problem to our problem.
\begin{figure}[h]
    \centering
   \includegraphics[width=0.85\textwidth]{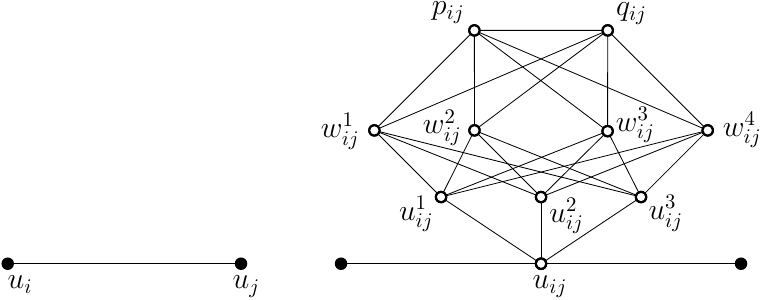}
    \caption{Gadget construction of $G_2$ from $G_1$, for example $5$-regular graph}
    \label{}
\end{figure}
Given a $d\geq 4$-regular graph $G_1 = (V(G_1), E(G_1))$ with vertex set $V(G_1) = \{v_1, v_2, \dots, v_n\}$, we construct a new graph $G_2$ as follows:
\begin{enumerate}
    \item First, we take the set of vertex $\{u_1, u_2, \dots, u_n\}$, where each $u_i \in V(G_2)$ corresponds to the vertex $v_i \in V(G_1)$.
       \item For each edge $(v_i,v_j)\in E(G_1)$, where $i < j$, we add a new vertex $u_{ij}$ and join with $u_i$ and $u_j$.
       \item For each $u_{ij}$, we introduce $(d-2)$ new vertices $u^1_{ij},u^2_{ij},\dots,u^{d-2}
       _{ij}$ and join with $u_{ij}$. Now, we add $(d-1)$ new vertices $w^{1}_{ij},w^{2}_{ij},\dots,w^{(d-1)}
       _{ij}$ and  $w^k_{ij}$ adjacent to $u^l_{ij}$ for each $k\in\{1,2,\dots,(d-1)\}$ and $l\in \{1,2,\dots,(d-2)\}$.
       \item For each $u_{ij}$, we introduce another two new vertices $p_{ij},q_{ij}$ and the edges $(p_{ij},w^k_{ij}),$ $(q_{ij},w^k_{ij}),(p_{ij},q_{ij})$ for each $k\in \{1,2,\dots,$ $(d-1)\}$.
        
       We denote the induced the $H_{ij}=G_2[\{u_i,u_j,u_{ij}\} \cup \{w_{ij}^k,u_{ij}^{l}:k=1,2\dots,d-1, l=1,2,\dots,d-2\}]$, $H_{ij}^0=G_2[\{u_i,u_j,u_{ij}\}]$, $H_{ij}^1=G_2[\{u_{ij}^{l}:l=1,2,\dots,d-2\}]$, $H_{ij}^2=G_2[\{w_{ij}^{k}:l=1,2,\dots,d-1\}]$, $H_{ij}^3=G_2[\{p_{ij},q_{ij}\}]$, therefore $H_{ij}=H_{ij}^0\cup H^1_{ij}\cup H^2_{ij}\cup H_{ij}^3$ .
      
\end{enumerate}
 \begin{lemma}\label{lemma:d regular hop}
           $G_1$ has a vertex cover of size at most $k$ if and only if $G_2$ has a $2$-step dominating set of size at most $k+2m$, where $m=|E(G_1)|$.
       \end{lemma}
       \begin{proof}
            Let us assume $S_{G_1}$ is a vertex cover of size at most $k$. Now we take $S_{G_2} =S_{G_1} \cup_{(v_i,v_j)\in E(G_1)} \{u_{ij},u_{ij}^1\}$, which is a $2$-step dominating set. Conversely, let $G_2$ have a $2$-step dominating set $S_{G_2}$ of size at most $k+2m$. Now to hop dominate $p_{ij}$ and $q_{ij}$ at least one vertices of $H_{ij}^1$ must be belongs to $S_{G_2}$. If $S_{G_2}\cap H_{ij}^1\geq 2$, then can update $S_{G_2}$ by $S_{G_2}\cup\{u_{ij}^1\}\setminus H_{ij}^1$. Now to hop dominate the vertices of $H_{ij}^2$ there exist a vertex $p\in H_{ij}^2$ or $p=u_{ij}$ must be in $S_{G_2}$. We can update $S_{G_2}$ by $S_{G_2}\cup\{u_{ij}\}\setminus H_{ij}^2$. Now to hop dominate $u_{ij}^1$ at least one of $u_i$ or $u_j$ must be in $S_{G_2}$. Notice that, the set $\{v_l:u_l\in S_{G_2}$\} is a vertex cover of $G_1$ of size $\leq k$. 
            \end{proof}
     Hence the \textsc{$2$-Step Domination} problem is NP-complete for $d$-regular graphs, where $d\geq4$.

       \section{Claw-free graph} \label{sec:clawfree}
       In this section, we prove that \textsc{Hop Domination} problem and \textsc{$2$-Step Domination} problem are NP-complete for claw-free graphs.
       \subsection{Hop Domination}
       \begin{figure}[h]
    \centering
   \includegraphics[width=0.9\textwidth]{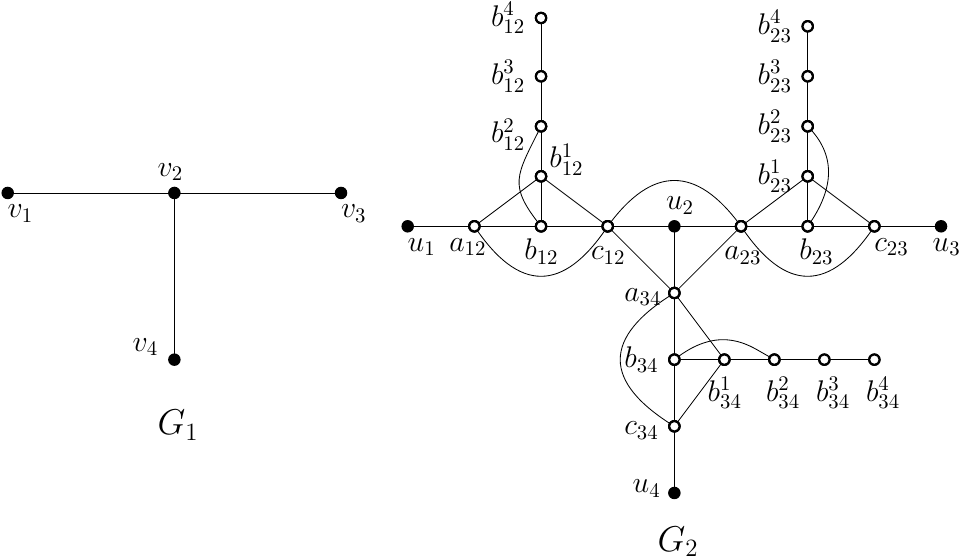}
    \caption{Gadget construction of $G_2$ from $G_1$}
    \label{}
\end{figure}
In this subsection, we show that \textsc{Hop Domination} problem in a claw-free graph is NP-complete, using the fact that the vertex cover problem \textsc{Np-Complete}\cite{alimonti2000some}. We reduce vertex cover problem to our problem.

       Consider any  graph $G_1$ with vertex set $V(G_1) = \{v_1,v_2,\dots,v_n\}$ and edge set $E(G_1)$.
We construct a new graph $G_2$ in the following way:
\begin{enumerate}
  \item  We add a new vertex set $\{ u_1,u_2,\dots,$ $u_n\}$ in $G_2$, where each $u_i$ corresponds to the vertex $v_i\in V(G_1)$. 
 \item For each edge $(v_i,v_j) \in E(G_1)$ with $i<j$, we draw a four length path $u_ia_{ij}b_{ij}c_{ij}u_j$ in $G_2$.  \item  In each $b_{ij}$ we draw a four length path $b_{ij}b^1_{ij}b^2_{ij}b^3_{ij}b^4_{ij}$. 
 \item Then we draw an edge between $b_{ij}$ and $b^2_{ij}$.
 \item For each edge $(v_i,v_j) \in E(G_1)$ with $i<j$,  we create a clique of size four with vertices $a_{ij},b_{ij},c_{ij},b^1_{ij}$. \item  Finally, for each $u_i\in G_2$, we form a clique with the vertex set is $N[u_i]$.
\end{enumerate}
\begin{lemma}
    
 $G_1$ has a vertex cover of size at most $k$ if and only if $G_2$ has a hop
dominating set of size at most $k + 2m$, where $m=|E(G_1)|$.
\end{lemma}
\begin{proof}

Suppose that, $G_1$ has a vertex cover, $S_{G_1}$ of size at most $k$. We consider the set $S_{G_2}=S_{G_1} \cup  \bigcup \{b^1_{ij},b^2_{ij}\}$, which is a hop dominating set.
Conversely, let $G_2$ has a hop dominating set, $S_{G_2}$ , of size at most $ k + 2m $.
If $| S_{G_2}\cap \{b_{ij},b^1_{ij},b^2_{ij},b^3_{ij}b^4_{ij}\}| \leq 1$ for some edge $e_{ij} \in E(G)$, then $b^3_{ij}$ or $b^4_{ij}$ 
is not hop dominated by $S_{G_2}$, a contradiction. Therefore, $|S_{G_2}|\cap \{b_{ij}b^1_{ij}b^2_{ij}b^3_{ij}b^4_{ij}\} | \geq 2$ for every edge $e_{ij} \in E(G_2)$. Now to hop dominate $b_{ij}^4$, $b_{ij}^2$ must be in  $S_{G_2}$. Similarly, to hop dominate $b_{ij}^3$, $b_{ij}^1$ or $b_{ij}$  must be in  $S_{G_2}$, w.l.o.g let $b_{ij}^1 \in S_{G_2}$. Now if any vertex $x\in N(u_i)$ belongs to $S_{G_2}$, $x$ can replaced by $u_i$. Now to hop dominate $b_{ij}$, $u_i$ or $u_j$ must be in $S_{G_2}$. Notice that, the set $\{v_l:u_l\in S_{G_2}$\} is a vertex cover of $G_1$ of size $\leq k$.
\end{proof}
Hence the \textsc{Hop Domination} problem is NP-complete for claw-free graphs.
\subsection{$2$-Step Domination}
\begin{figure}[h]
    \centering
   \includegraphics[width=0.85\textwidth]{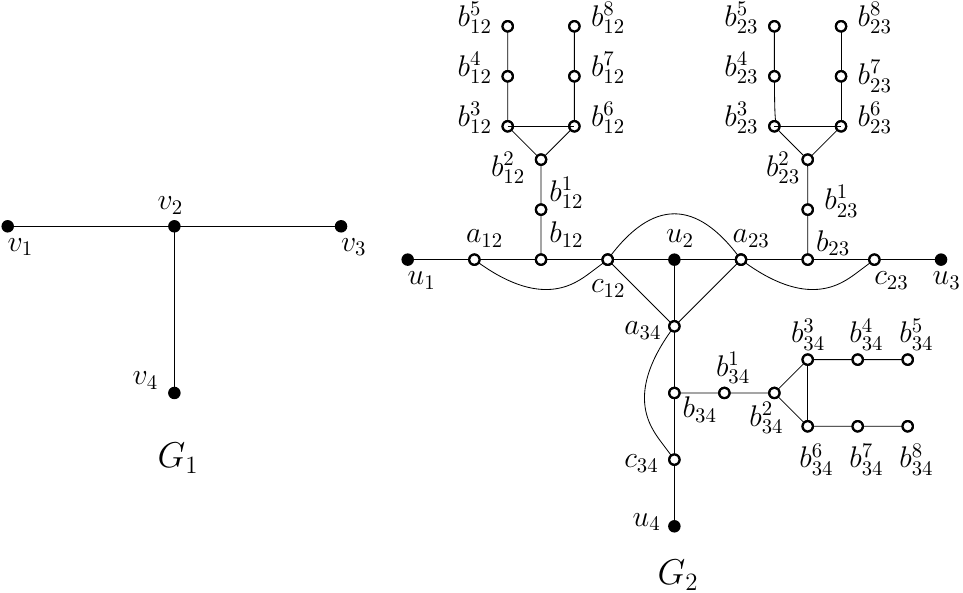}
    \caption{Gadget construction of $G_2$ from $G_1$}
    \label{}
\end{figure}
In this subsection, we show that \textsc{$2$-Step Domination} problem in a claw-free graph is NP-complete, using the fact that the vertex cover problem \textsc{Np-Complete}\cite{alimonti2000some}. We reduce vertex cover problem to our problem.

 Consider any  graph $G_1$ with vertex set $V(G_1) = \{v_1,v_2,\dots,v_n\}$ and edge set $E(G_1)$.
We construct a new graph $G_2$ in the following way:
 \begin{enumerate}
   
 \item  We add a new vertex set $\{ u_1,u_2,\dots,$ $u_n\}$ in $G_2$, where each $u_i$ corresponds to the vertex $v_i\in V(G_1)$. 
 \item  For each edge $(v_i,v_j) \in E(G_1)$ with $i<j$, we draw a four length path $u_ia_{ij}b_{ij}c_{ij}u_j$ in $G_2$.  \item  In each $b_{ij}$ we draw a two length path $b_{ij}b^1_{ij}b^2_{ij}$. 
 \item For each $b_{ij}^2$, we introduce six new vertices $b_{ij}^3,b_{ij}^4,b_{ij}^5,b_{ij}^6,b_{ij}^7,b_{ij}^8$, and we draw the edges $(b_{ij}^3,b_{ij}^4),(b_{ij}^4,b_{ij}^5),(b_{ij}^3,b_{ij}^6),(b_{ij}^6,b_{ij}^7),(b_{ij}^7,b_{ij}^8),(b_{ij}^2,b_{ij}^3),(b_{ij}^2,b_{ij}^6)$.
 \item Finally, for each $u_i\in G_2$, we form a clique with the vertex set is $N[u_i]$. 
 \end{enumerate}
\begin{lemma}
$G_1$ has a vertex cover of size at most $k$ if and only if $G_2$ has a $2$-step
dominating set of size at most $k + 4m$, where $m=|E(G_1)|$.
\end{lemma}
\begin{proof}
    
Suppose that $G_1$ has a vertex cover, $S_{G_1}$ of size at most $k$. We consider the set $S_{G_2}=S_{G_1} \cup  \bigcup \{b^3_{ij},b^6_{ij},b_{ij},b_{ij}^1\}$, which is a $2$-step dominating set.
Conversely, let $G_2$ have a $2$-step dominating set, $S_{G_2}$ of size at most $ k + 4m $. $| S_{G_2}\cap \{b_{ij},b^1_{ij},b^2_{ij},b^3_{ij},b^4_{ij},b^6_{ij},b^7_{ij}\}| \geq 4$ for some edge $e_{ij} \in E(G)$. If possible let $| S_{G_2}\cap \{b_{ij},b^1_{ij},b^2_{ij},b^3_{ij},b^4_{ij},b^6_{ij},b^7_{ij}\}| \leq 3$ for some edge $e_{ij} \in E(G)$. Now to hop dominate $b_{ij}^5$ and $b_{ij}^8$  vertices $b_{ij}^3$ and $b_{ij}^6$ must be in $S_{G_2}$. Now to hop dominate $b_{ij}^3$ and $b_{ij}^6$ either $b_{ij}^1$ be in $S_{G_2}$ or one of $\{b_{ij}^5,b_{ij}^4\}$,$\{b_{ij}^5,b_{ij}^8\}$,$\{b_{ij}^6,b_{ij}^7\},\{b_{ij}^4,b_{ij}^7\}$ must be in $S_{G_2}$. Now if $ S_{G_2}\cap \{b_{ij},b^1_{ij},b^2_{ij},b^3_{ij},b^4_{ij},b^6_{ij},b^7_{ij}\}=\{b^3_{ij},b^6_{ij},b^1_{ij}\}$, then $b^2_{ij}$ can not hop dominated by $S_{G_2}$. Hence $| S_{G_2}\cap \{b_{ij},b^1_{ij},b^2_{ij},b^3_{ij},b^4_{ij},b^6_{ij},b^7_{ij}\}| \geq 4$. We can select $S_{G_2}$ in such a way that $ S_{G_2}\cap \{b_{ij},b^1_{ij},b^2_{ij},b^3_{ij},b^4_{ij},b^6_{ij},b^7_{ij}\}=\{b^3_{ij},b^6_{ij},b^1_{ij},b_{ij}\}$. Now if any vertex $x\in N(u_i)$ belongs to $S_{G_2}$, $x$ can replaced by $u_i$. Now to hop dominate $b_{ij}$, $u_i$ or $u_j$ must be in $S_{G_2}$. Notice that, the set $\{v_l:u_l\in S_{G_2}$\} is a vertex cover of $G_1$ of size $\leq k$. 
\end{proof}
Hence the \textsc{$2$-Step Domination} problem is NP-complete for claw-free graphs.

\section{Conclusion}\label{sec:conclusion}

In this paper, we studied the computational complexity of the \textsc{Hop Domination} and \textsc{$2$-Step Domination} problems. We proved that both problems are NP-complete on several important graph classes, including unit disk graphs, regular graphs ($d$-regular, for every $d\geq 3$), and claw-free graphs. These results strengthen the known complexity landscape of exact distance domination problems and demonstrate that the problems remain computationally intractable even under strong structural restrictions on the input graphs.

Several interesting questions remain open for future research. It would be worthwhile to study the parameterized and approximation complexity of these problems on restricted graph classes such as unit disk,  planar, regular and claw-free  graphs. Further study of kernelization and fixed-parameter tractability for these problems may also lead to deeper insights into the structure of exact distance domination problems.

 \bibliographystyle{abbrv}
\bibliography{ref}
\end{document}